\documentclass[journal]{IEEEtran}

\usepackage[english]{babel}
\usepackage[utf8x]{inputenc}
\usepackage{amsmath,amsthm}

\usepackage{multirow,booktabs}

\usepackage{graphicx}

\usepackage{amsfonts}
\graphicspath{{images/}}
\usepackage{subfigure}
\usepackage[justification=centering]{caption}
\usepackage{epstopdf}
\usepackage{bm}
\usepackage[flushleft]{threeparttable}

\newtheorem{assumption}{Assumption}

\newtheorem{theorem}{Theorem}
\newtheorem{remark}{Remark}

\usepackage{upgreek}
\usepackage{fixmath}

\usepackage{cite}
\usepackage{xcolor}
\allowdisplaybreaks

\usepackage{balance}
\usepackage{misc}

\usepackage{xspace}

\newcommand{\sat}{{\mathrm{sat}}\xspace}
\newcommand{\T}{\top}
\newcommand{\rmd}{{\mathrm d}}

\graphicspath{{Pics/}}

\newtheorem{thm}{Theorem}
\newtheorem{cor}{Corollary}

\newtheorem{rmk}{Remark}

\begin{document}

\title{Robust Extended Kalman Filtering for Systems with Measurement Outliers}
\author{Huazhen Fang, Mulugeta A. Haile and Yebin Wang
\thanks{H. Fang is with Department of Mechanical Engineering, University of Kansas, Lawrence, KS 66045, USA (e-mail: fang@ku.edu) . }
\thanks{M. A. Haile is with Vehicle Technology Directorate, US Army Research Laboratory, Aberdeen, MD 21005, USA (e-mail: mulugeta.a.haile.civ@mail.mil).}
\thanks{Y. Wang is with Mitsubishi Electric Research Laboratories, Cambridge, MA 02139, USA (e-mail: yebinwang@ieee.org).}
}

\maketitle

\begin{abstract}
Outliers can contaminate the measurement process of many nonlinear  systems, which can be caused by sensor errors, model uncertainties, change in ambient environment, data loss or malicious cyber attacks. When the   extended Kalman filter (EKF) is applied to such systems for state estimation, the outliers can seriously reduce the estimation accuracy. This paper proposes an innovation saturation mechanism to modify the EKF toward building robustness against outliers. This mechanism applies a saturation function to the innovation process that the EKF leverages to correct the state estimation. As such, when an outlier occurs, the distorting innovation is saturated and thus prevented from damaging the state estimation. The mechanism features an adaptive adjustment of the saturation bound.  The design leads to the development robust EKF approaches for continuous- and discrete-time systems.  They are proven to be capable of generating bounded-error estimation in the presence of bounded outlier disturbances. An application study about mobile robot localization is presented, with the
numerical simulation  showing the efficacy of the proposed design.  Compared to existing methods, the proposed approaches can effectively reject outliers of various magnitudes, types and   durations,  at significant computational efficiency and without requiring measurement redundancy. 

\end{abstract}

\begin{IEEEkeywords}
Kalman filter, extended Kalman filter, robust estimation, measurement outlier, localization.
\end{IEEEkeywords}

\IEEEpeerreviewmaketitle

\section{Introduction}

The Kalman filter (KF) is arguably the most celebrated estimation technique in the literature, which can optimally estimate  the state of a linear dynamic system on the basis of a model and a stream of noisy measurements.  In practice, it is the extended KF (EKF), the KF's nonlinear version, that is the most widely used, because   real-world systems usually involve nonlinearities~\cite{Fang:JAS:2018}. The EKF's applications  range from  control systems to signal processing, system health monitoring, navigation and econometrics.  However, a major challenge for high-quality estimation in practice is the   measurement outliers, which can come from a diversity of sources, e.g., unreliable sensors, environmental variability, data dropouts in transmission, channel biases, incorrect assumptions about noises, model mismatch, and data falsification attacks from cyberspace~\cite{Ting:IROS:2007,Zhang:TCT:2010,Haile:TR:2016,Sinopoli:TAC:2004,Jia:TCST:2013,Li:TCNS:2017}. When outliers corrupt the measurement data, the performance of the EKF can be seriously degraded or damaged, and as such without adequate robustifcaition against outliers, it will not be viable for real-world practical estimation problems.


{\em Literature review}. Robust state estimation against measurement outliers has attracted significant interest from researchers during the past years. A majority of the effort has been devoted to robustifying the standard linear KF.  One can divide the existing methods mainly into three main categories. The first category models the measurement noises using  heavy-tailed distributions rather than exponential distributions, e.g., the Gaussian distribution as often assumed in the classical KF, to capture the occurrence of an outlier. Heavy-tailed Gaussian-mixture~\cite{Masreliez:TAC:1975,Sorenson:AUTO:1971} and $t$-distributed noise models~\cite{Meinhold:JASA:1989} are used to modify the KF for better robustness. In~\cite{Sarkka:TAC:2009}, an outlier is viewed as a result of measurement noises with variable covariances. Assuming the noise covariances to follow an inverse Gamma distribution, it proposes  to add to the KF a procedure of adaptive identification of the key parameters involved in the inverse Gamma distribution. 
 Methods of the second category seek to assign the measurement  at each time instant with a weight, in an attempt to downweight outlying measurements. In~\cite{Ting:IROS:2007}, an expectation-maximization (EM) algorithm is used to enable adaptive determination of the weight for a measurement. The results are generalized in~\cite{Agamennoni:ICRA:2011} and extended to the smoothing problem. In~\cite{Gandhi:TSP:2010,Palma:IJACSP:2017}, a measurement-weighting-based prewhitening procedure  is designed to decorrelate outliers from normal measurements as a basis for building a robust KF. The third category of methods extends the KF to conduct simultaneous state and input estimation,  regarding an outlier as an input added to  the measurement and estimating it together with the state. To accomplish this, a few methods have been built on minimum variance unbiased estimation~\cite{Keller:AUTO:1997,Hsieh:TAC:2000, Gillijns:AUTO:2007,Fang:IJACSP:2011,Yong:AUTO:2016,Shi:AUTO:2016}. Further, Bayesian methods are developed in~\cite{Fang:AUTO:2013} to achieve joint state and outlier estimation from the perspective of probabilistic filtering.

Although much attention has been given to outlier rejection for the linear KF, it is more  important and pressing  to robustify the EKF due to its practical significance. The  operation of the EKF at every time instant relies on  linearized approximations based on the   estimation in the previous time instant. When   measurement outliers arise, they would increase the state estimation error, which in turn will amplify the error involved in linearized approximations. This may further drive away the estimation at the next time instant, potentially leading to divergence. Hence, the EKF is more vulnerable to outliers, making it in an urgent need for outlier rejection. However, the current literature includes very few studies in this regard, due to the associated difficulty.
 In~\cite{Mu:JEM:2015}, the EKF is blended with a procedure that detects an outlier by evaluating the probability of its occurrence   based on innovation statistics. The method in~\cite{Gandhi:TSP:2010} is modified in~\cite{Gandhi:VT:2009} to deal with outliers affecting the EKF run.

It is noteworthy that the above robust KF/EKF techniques, despite their importance, generally involve a dramatic  increase in computational complexity, due to iterative optimization or other computationally expensive procedures to detect or suppress outliers at every time instant. Besides, some of them  require measurement redundancy to differentiate outliers from normal measurements or estimate them directly. This, however, is not always possible, because a real system often allows only a limited number of sensors to be deployed.

In addition to the above outlier-robust KF/EKF methods, one can also find   other types of estimation approaches in the literature aimed at suppressing outliers. Among them is the well-known   $\mathcal{H}_\infty$ filtering~\cite{Simon:2006,Li:AUTO:2016}, which considers outliers as unknown yet bounded  uncertainty. Yet, this approach introduces   conservatism as it performs worst-case estimation by design. A stubborn observer is developed in~\cite{Alessandri:AUTO:2018}, which employs a saturation function in the output injection signal to mitigate the influence of outliers. This method is not only computationally fast but also can deal with very large outliers. Nonetheless, it is applicable to only linear   systems suffering outliers that occur occasionally and individually. 

{\em Statement of contribution}. In this work, we offer an interesting new design to enhance the robustness of the EKF against measurement outliers, presenting a three-fold contribution.  First, we propose a unique innovation saturation mechanism to reject  outliers and ensure the performance of the EKF.  The innovation plays a key role in correcting the state prediction in the EKF but can be distorted by outliers. To overcome such a vulnerability, our mechanism saturates the innovation when it is unreasonably large in order  to reduce the effects of outliers. At the core of the mechanism is a procedure for adaptively adjusting the saturation bound  to effectively grasp the change of the innovation. Along this line, we develop the innovation-saturated EKF (IS-EKF)  for both continuous- and discrete-time systems. It is noteworthy that  the innovation saturation has been considered in~\cite{Alessandri:AUTO:2018}, which, however,  can only suppress outliers that appear singly or one at a time. By contrast, our design can handle outliers that persist for a  long period. Second, we analyze the stability of the proposed IS-EKF for the linear case, proving that it produces bounded-error estimation under certain conditions when outlier disturbances are bounded. Finally, we apply the proposed IS-EKF to the problem of mobile robot localization, demonstrating its effectiveness in providing reliable estimation in the presence of GPS outliers. Compared to the existing methods, the proposed IS-EKF approaches are structurally concise, computationally efficient, free from requiring measurement redundancy, and robust against   outliers of different magnitudes, types and durations, thus lending itself well practical application. This work is a significant extension of our previous work~\cite{Fang:CDC:2018}, in which a simple saturation mechanism is designed to robustify the linear KF. 

{\em Organization.}This paper is organized as follows. Section~\ref{SKF-CT-Sec} develops the IS-EKF for nonlinear continuous-time systems and analyzes its stability. Section~\ref{SKF-DT-Sec} extends the results to   discrete-time systems. An application example based on mobile robot localization is provided in Section~\ref{simulation} to illustrate the usefulness of the proposed design. Finally, our concluding remarks are presented in Section~\ref{conclusion}.

{\em Notation:} Notations used throughout this paper are standard. 
The $n$-dimensional Euclidean space is denoted as $\mathbb{R}^n$. For a vector, $\|\cdot\|$ denotes its 2-norm. The notation  $I$ is an identity matrix; $X>0$ ($\geq 0$) means that $X$ is a real, symmetric and positive definite (semidefinite) matrix; for a symmetric block matrix, we use a star ($\star$) to represent a symmetry-induced block; the notation $\mathrm{diag}(\ldots)$  stands for a block-diagonal matrix. The minimum and maximum eigenvalues of a real, symmetric matrix are denoted by $\underline{\lambda} (\cdot)$ and $\bar{\lambda} (\cdot)$, respectively.  Matrices are assumed to be
compatible for algebraic operations if their dimensions are not explicitly stated.

\section{IS-EKF for   Continuous-Time Systems}\label{SKF-CT-Sec}

This section develops the IS-EKF approach for a nonlinear continuous-time system and then offers analysis of  its stability.

\subsection{IS-EKF Architecture}
Consider the following model
\begin{equation}\label{CT_system}
\left\{
\begin{aligned}
\dot x_t &= f(x_t)+w_t,\\
y_t &= h(x_t)+Dd_t+v_t,
\end{aligned}
\right.
\end{equation}
where $x \in {\mathbb R}^n$ is the state vector,  $y \in {\mathbb R}^p$ the measurement vector, and $w_t\in {\mathbb R}^n$ and $v_t\in {\mathbb R}^p$  zero-mean, mutually independent noises with covariances given by $Q\geq 0$ and $R>0$, respectively. The nonlinear mappings $f$ and $h$ represent the  state evolution and measurement functions, respectively. Note that the measurement $y_t$ is subjected to the outlier effects caused by an unknown disturbance $d_t \in {\mathbb R}^m$. The matrix $D$ shows the relation between $d_t$ and $y_t$ and is assumed to be unknown.  

Modifying the conventional EKF, we propose the the following IS-EKF procedure:
\begin{subequations}\label{SKF-CT-Outline}
\begin{align}\label{state-update}
\dot{\hat x}_t &= f( \hat x_t) +K_t \cdot \sat_\sigma \left(y_t - h\left( \hat x_t \right) \right),\\
K_t &=  P_t C_t^\T R^{-1},\\ \label{P-CT-update}
\dot P_t &= A_t P_t+P_tA_t^\T+Q-K_t R K_t^\T,
\end{align}
\end{subequations}
where $K_t$ is the estimation gain matrix, and $P_t$ is a positive definite matrix that approximately represents the estimation error covariance in the standard EKF, and
\begin{align*}
A_t = \left.\frac{\partial f}{\partial x}\right|_{\hat x_t}, \ \ C_t = \left.\frac{\partial h}{\partial x}\right|_{\hat x_t}.
\end{align*}
Note that,
for the conventional EKF, the state estimation is corrected by the innovation $\left(y_t - h(\hat x_t) \right)$. Its effectiveness, however, can be compromised if $y_t$ is corrupted by an outlier. To address this issue, we  use a saturated innovation instead, as shown in~\eqref{state-update}. Specifically, it is defined as
\begin{align}\label{saturated-innovation}
\sat_\sigma \left(y_t - h( \hat x_t)\right)  = \left[
\begin{matrix}
\vdots \cr \sat_{\sqrt{\sigma_i}} \left(y_{i,t} - h_i ( \hat x_t)\right) \cr \vdots
\end{matrix}
\right],
\end{align}
where $\sigma_i>0$, $y_i$ is the $i$-th element of $y$, and $h_i$  the $i$-th element of $h$. For a variable $r$, the saturation function is defined as $\sat_{\epsilon}(r) = \max\left\{ - \epsilon, \min\{ \epsilon, r\}\right\}$.  For~\eqref{saturated-innovation}, the saturation range $[-\sqrt{\sigma_i},\sqrt{\sigma_i}]$ can be loosely viewed as an anticipated range of the innovation. If falling within this range, the innovation is considered as reasonable and applied without change to update the state estimation. Otherwise, it may be affected by an outlier and thus saturated to prevent the outlier from dragging the estimation away from a correct course.

For the EKF, it is observed that a fixed saturation bound will be too limited in the efficacy of rejecting outliers, as it may either confuse with an outlier a certain measurement generating a large innovation or miss an outlier approximately falling within the bounded range. More often than not, one  can also find it practically difficult  to select a fixed bound, especially when knowledge  of the outliers is scarce. We thus propose the following procedure  to adaptively adjust the saturation bound:
\begin{subequations}\label{sat_bound_dynamics-CT}
\begin{align}\label{sat_bound_dynamics-CT-1}
\dot \sigma_{i,t} &= \lambda_{1,i} \sigma_{i,t} + \gamma_{1,i} \varepsilon_{i,t} e^{-\varepsilon_{i,t}},  \ \ \sigma_{i,0}>0,\\  \label{sat_bound_dynamics-CT-2}
\dot \varepsilon_{i,t} &= \lambda_{2,i} \varepsilon_{i,t} + \gamma_{2,i} \left(y_{i,t} - h_i (\hat x_t) \right)^2,  \ \ \varepsilon_{i,0}>0,
\end{align}
\end{subequations}
for $i=1,2,\ldots,p$, where $\lambda_{1,i},\lambda_{2,i}<0$ and $\gamma_{1,i},\gamma_{2,i}>0$. For convenience of notation, we  define 
\begin{align*}
\Lambda_i = \mathrm{diag}([ \cdots \ \lambda_{i,j} \ \cdots ]), \
\Gamma_i = \mathrm{diag}([ \cdots \ \gamma_{i,j} \ \cdots ]),
\end{align*} for $i=1,2$ and $j=1,2,\cdots,p$.

Based on~\eqref{sat_bound_dynamics-CT},  the innovation saturation bound $\sigma_i$ will dynamically change according to the innovation $\left(y_{i,t} - h_i (\hat x_t) \right)$.  This mechanism specifically builds on a double-layer structure. The lower layer,   based on~\eqref{sat_bound_dynamics-CT-2},  tracks the   changes in the innovation signal --- the variable $\varepsilon_i$ will be maintained at an appropriate level when the innovation is normal but become large when the innovation is altered by outliers. The upper layer, based on~\eqref{sat_bound_dynamics-CT-1}, is designed to respond to the  change in innovation by adjusting the saturation bound. As is seen, $\sigma_i$ will rapidly decrease  when $\varepsilon_i$ is too large but lie within a reasonable range when given a normal $\varepsilon_i$. By adapting the saturation bound, the proposed design  in~\eqref{sat_bound_dynamics-CT}   will enable an improved discernment between an outlier and a normal measurement.



\subsection{Stability Analysis}

While there exist some  studies, it has been widely acknowledged as a challenge to determine the exact conditions for the asymptotic stability of the EKF. It is much more difficult, if not impossible, to do this for the IS-EKF, because of the added innovation saturation procedure and the nonlinear update of the saturation bound. To formulate a tractable analysis, we restrict our attention to the asymptotic stability of the IS-EKF for a linear deterministic system:
\begin{align*}
\left\{
\begin{aligned}
\dot x_t &=A x_t,\\
y_t &= Cx_t+Dd_t,
\end{aligned}
\right.
\end{align*}
for which the IS-EKF acts as a state observer, with the estimation performed by the innovation-saturated KF. Here, we assume that $(A,Q^{1\over 2})$ is stabilizable and that $(A,C)$ is detectable, as often needed for estimation. 

Let us first define the state estimation error as $e_t = \hat x_t - x_t$. The dynamics of $e_t$ is governed by
\begin{align}\label{e_t_dynamics}
\dot e_t = A e_t - K_t \cdot \sat_\sigma(Ce_t-Dd_t).
\end{align}
To proceed further, we define the following matrix
\begin{align*}\small
S_t =\left[ \begin{matrix} M_t-\alpha P_t^{-1} & -C^\T \left( R^{-1} +W \right) &  C^\T  (\Gamma_2-R^{-1}) D \cr
\star & 2W  & WD \cr
\star & \star & U
\end{matrix} \right],
\end{align*}
where  
$M_t= P_t^{-1}QP_t^{-1} + C^\T(R^{-1}-\Gamma_2)C $, $W$ is a diagonal positive definite matrix, $U$ a positive definite matrix, and $\alpha>0$ a positive scalar. Furthermore, we recall a well-known fact~\cite{Kucera:Kybernetika:1973}: if $(A,Q^{1\over 2})$ is stabilizable and $(A,C)$ detectable, $P_t$ for $P_0\geq 0$ in~\eqref{P-CT-update} will approach a unique positive-definite solution $ P_\infty$ satisfying
\begin{align*}
A  P_\infty +  P_\infty A^\T+Q- P_\infty C^\T R^{-1} C  P_\infty=0 .
\end{align*}

The following result can be obtained regarding the stability of the proposed IS-EKF for the linear deterministic case.
\begin{thm}\label{SKF-Stability-Thm}
Suppose    $\|d_t\|\leq 
\mu<\infty$ and $\rho = e^{-1}\sum_i \gamma_{1,i} < \infty$, where  $\mu, \rho>0$.
If there exist $P_0$, $W$, $U$, $\alpha$ and $\Gamma$ such that $S_t\geq 0$ and   $0 < \alpha \leq - \max \{\ldots, \lambda_{1,i},\lambda_{2,i},\ldots\} $ for $i=1,2,\ldots,p$, then the estimation error $e_t$ is upper bounded with
\begin{subequations}
\begin{align}\label{CT-SKF-error-bound-1}
\|e_t\| &\leq \sqrt{\frac{1}{c_2} \left[ e^{-\alpha t} V_0 + \frac{1}{\alpha}(1 - e^{-\alpha t})   ( c_1 \mu^2+\rho)\right]},\\\label{CT-SKF-error-bound-2}
\lim_{t\rightarrow\infty} \|e_t\| &\leq \sqrt{\frac{c_1 \mu^2+\rho  }{\alpha c_3}},
\end{align}
\end{subequations}
where $c_1 = \bar \lambda {(U+D^\T \Gamma_2 D)} $, $c_2 = \underline{\lambda}(P_t^{-1})$ and $c_3 = \underline{\lambda}(P_\infty^{-1})$.
\end{thm}

{\em Proof:} We consider using the Lyapunov function 
\begin{align*}
V_t = e_t^\T P_t^{-1} e_t + \textstyle\sum_i \sigma_{i,t} + \textstyle\sum_i \varepsilon_{i,t} .
\end{align*}
The first-order time derivative of $V_t$ along~\eqref{e_t_dynamics} is
\begin{align*}
\dot V_t &= 2 e_t^\T P_t^{-1} \dot e_t + e_t^\T {\rmd \left(P_t^{-1}\right) \over \rmd t}  e_t + \sum_i \dot \sigma_{i,t} + \sum_i \dot \varepsilon_{i,t}\\
&= 2 e_t^\T P_t^{-1} \left[A e_t - K_t  \cdot\sat_\sigma\left(Ce_t-Dd_t\right)\right]\\
&\quad-e_t^\T P_t^{-1}  \left(AP_t+P_tA^\T+Q-K_t R K_t^\T\right) P_t^{-1}e_t\\
&\quad + (Ce_t-Dd_t)^\T \Gamma_2 (Ce_t-Dd_t) +\textstyle \sum_i\lambda_{1,i} \sigma_{i,t} \\
&\quad +\textstyle \sum_i\lambda_{2,i} \varepsilon_{i,t} +\textstyle \sum_i\gamma_{1,i} \varepsilon_{i,t} e^{-\varepsilon_{i,t}} \\
&\leq-e_t^\T P_t^{-1} Q P_t^{-1} e_t +e_t^\T C^\T (R^{-1}+\Gamma)Ce_t \\&\quad-2e_t^\T C^\T R^{-1} \sat_\sigma\left(Ce_t-Dd_t\right) - 2 e_t^\T C^\T \Gamma_2 D d_t \\ 
&\quad +d_t^\T D^\T \Gamma_2 D d_t+\textstyle\sum_i\lambda_{1,i} \sigma_{i,t } +\textstyle\sum_i\lambda_{2,i} \varepsilon_{i,t } + \rho,
\end{align*}
where the relation $\varepsilon_{i,t} e^{-\varepsilon_{i,t}} \leq e^{-1}$ is used. 
Let us define $s_t = Ce_t - Dd_t-\sat_\sigma(Ce_t-Dd_t)$. Then,
\begin{align*}
\dot V_t &\leq - e_t^\T M_t e_t  +2 e_t^\T C^\T R^{-1} s_t + 2e_t^\T C^\T  (R^{-1}-\Gamma_2) D d_t \\ 
&\quad  +d_t^\T D^\T \Gamma_2 D d_t+\textstyle\sum_i\lambda_{1,i} \sigma_{i,t } +\textstyle\sum_i\lambda_{2,i} \varepsilon_{i,t }+\rho.
\end{align*}
By~\cite[Lemma 1.6]{Tarbouriech:Springer:2011}, we have
\begin{align*}
-s_t^\T W \left(s_t -  Ce_t+Dd_t \right) \geq 0.
\end{align*}
It then follows that
\begin{align*}
&\dot V_t \leq \dot V_t - 2 s_t^\T W \left(s_t -  Ce_t+Dd_t \right) \\
& \leq - e_t^\T M_t e_t  +2 e_t^\T C^\T \left( R^{-1} +W \right)s_t  - 2s_t^\T W s_t \\
&\quad+ 2e_t^\T C^\T  (R^{-1}-\Gamma_2) D d_t -2s_t^\T WDd_t+ \\ 
&\quad +d_t^\T D^\T \Gamma_2 D d_t+\textstyle\sum_i\lambda_{1,i} \sigma_{i,t} +\textstyle\sum_i\lambda_{2,i} \varepsilon_{i,t }  + \rho\\
&= - \left[ \begin{matrix} e_t \cr s_t \cr d_t \end{matrix} \right]^\T  
\left[ \begin{matrix} M_t & -C^\T \left( R^{-1} +W \right) &  C^\T  (\Gamma_2-R^{-1}) D \cr
\star & 2W  & WD \cr
\star & \star & U
\end{matrix} \right] \\
& \quad \cdot \left[ \begin{matrix} e_t \cr s_t \cr d_t \end{matrix} \right]+d_t^\T \left(U+D^\T\Gamma_2 D \right)d_t+\textstyle\sum_i\lambda_{1,i} \sigma_{i,t} +\textstyle\sum_i\lambda_{2,i} \varepsilon_{i,t }\\ &\quad  + \rho.
\end{align*}
If $S_t\geq 0$, we have
\begin{align*}
\dot V(t)&\leq -\alpha e_t^\T P_t^{-1} e_t - \alpha \textstyle \sum_i \sigma_{i,t} - \alpha \textstyle \sum_i \varepsilon_{i,t} 
\\ &\quad +  d_t^\T \left(U+D^\T\Gamma_2 D \right)d_t  +  \textstyle \sum_i (\lambda_{1,i}+\alpha)\sigma_{i,t} \\
& \quad  +  \textstyle \sum_i (\lambda_{2,i}+\alpha)\varepsilon_{i,t}  + \rho.
\end{align*}
If $0<\alpha \leq - \max(\lambda_i) $, one has $\lambda_i+\alpha \leq 0$. Then,
\begin{align*}
\dot V_t&\leq -\alpha V_t + d_t^\T \left(U+D^\T\Gamma_2 D \right)d_t +\rho\\
&\leq -\alpha V_t  +c_1 \|d_t\|^2 +\rho \\
&\leq -\alpha V_t + c_1 \mu^2+\rho.
\end{align*}
Hence,
\begin{align*}
V_t \leq e^{-\alpha t} V_0 + \frac{1}{\alpha} (1 - e^{-\alpha t})   (c_1\mu^2+\rho).
\end{align*}
Furthermore, $V_t\geq c_2 \|e_t\|^2$. Then,
\begin{align*}
\| e_t\|^2 \leq \frac{1}{c_2} \left[ e^{-\alpha t} V_0 + \frac{1}{\alpha}(1 - e^{-\alpha t})  ( c_1 \mu^2+\rho) \right],
\end{align*}
which implies~\eqref{CT-SKF-error-bound-1}. When $t\rightarrow \infty$, we can obtain~\eqref{CT-SKF-error-bound-2}. \hfill$\bullet$

Theorem~\ref{SKF-Stability-Thm} shows that, when applied to a noise-free linear system with upper-bounded outliers,  the proposed IS-EKF scheme can lead to bounded-error estimation if some conditions are satisfied. Further,  the following corollary can be developed.
\begin{cor}\label{SKF-Stability-Cor}
Suppose    $\|d_t\|\leq 
\mu<\infty$ for $\mu >0$.
If there exist $P_0$, $W$, $U$, $\alpha$ and $\Gamma$ such that $S_t \geq 0$ and   $0 < \alpha \leq - \max \{\ldots, \lambda_{1,i},\lambda_{2,i}+\gamma_{1,i,}\ldots\} $ for $i=1,2,\ldots,p$, then the estimation error $e_t$ is upper bounded with
\begin{subequations}
\begin{align}\label{CT-SKF-error-bound-1}
\|e_t\| &\leq \sqrt{\frac{1}{c_2} \left[ e^{-\alpha t} V_0 + \frac{1}{\alpha}(1 - e^{-\alpha t})    c_1 \mu^2 \right]},\\\label{CT-SKF-error-bound-2}
\lim_{t\rightarrow\infty} \|e_t\| &\leq \sqrt{\frac{c_1   }{\alpha c_3}}\mu.
\end{align}
\end{subequations}
\end{cor}
To prove this result,  one can generally follow the proof of Theorem~\ref{SKF-Stability-Thm} while upper-bounding $\varepsilon_{i,t} e^{-\varepsilon_{i,t}} $ by $\varepsilon_{i,t}$  instead of $\leq e^{-1}$.  Corollary~\ref{SKF-Stability-Cor} implies that $e_t$ will exponentially approach zero as $t \rightarrow \infty$ if $\lim_{t\rightarrow \infty }d_t=0$. 


\begin{rmk}
According to Theorem~\ref{SKF-Stability-Thm} and Corollary~\ref{SKF-Stability-Cor}, selection of $P_0$ can be important for making the condition  $S_t \geq 0$ satisfied. Given the structure of $S_t$, it is cautioned that too large a $P_0$ may bring the risk of  divergent estimation. However,  $P_t$ is monotonically non-decreasing for $P_0 = 0$ if $(A,Q^{1\over 2})$ is stabilizable and $(A,C)$ detectable~\cite{Kucera:Kybernetika:1973}. Leveraging this property, it is suggested  that $P_0$ be set to  be a small number (close to zero in particular if some accurate prior knowledge about the initial state is available) when the IS-EKF approach is to be implemented. \hfill$\bullet$
\end{rmk}

\section{IS-EKF for Discrete-Time Systems}\label{SKF-DT-Sec}

Extending the notion in Section~\ref{SKF-CT-Sec}, this section investigates the development of IS-EKF for nonlinear discrete-time  systems. 

\subsection{IS-EKF Architecture}
Consider a nonlinear discrete-time model
\begin{equation}\label{DT_system}
\left\{
\begin{aligned}
 x_{k+1} &= f( x_k) +w_k,\\
y_k &= h(x_k)+Dd_k+v_k.
\end{aligned}
\right.
\end{equation}
The notations in above are the same as in Section~\ref{SKF-CT-Sec}. Still, the  noises $w_k$ and $v_k$ are zero-mean, mutually independent with covariances $Q\geq0$ and $R>0$, respectively. 

For this system, we propose the IS-EKF   as follows:
\begin{subequations}\label{SKF-DT-Outline}
\begin{align}
\hat x_{k|k-1} &= f\left( \hat x_{k-1|k-1} \right),\\
P_{k|k-1} &= A_{k-1} P_{k-1|k-1} A_{k-1}^\T+Q,\\ \label{DT-saturated-update}
\hat x_{k|k} &= \hat x_{k|k-1} +K_k \cdot \sat_\sigma \left(y_k - h\left( \hat x_{k|k-1}\right) \right),\\
K_k &=  P_{k|k-1} C_k^\T \left( C_k P_{k|k-1}C_k^\T+R \right)^{-1},\\ 
P_{k|k} &= P_{k|k-1}-K_k \left( C_kP_{k|k-1}C_k^\T+R \right) K_k^\T,
\end{align}
\end{subequations}
where $\hat x_{k|k-1}$ is the one-step-forward prediction of $x_k$,   $\hat x_{k|k}$ the updated estimate when $y_k$ arrives to correct the prediction, and
\begin{align*}
A_k = \left.\frac{\partial f}{\partial x}\right|_{\hat x_{k|k}}, \ \ C_k = \left.\frac{\partial h}{\partial x}\right|_{\hat x_{k|k-1}}.
\end{align*}
In addition, $K_k$ is the estimation gain, and $P_{k|k-1}$ and $P_{k|k}$ the approximate estimation error covariances in the standard EKF. Akin to~\eqref{SKF-CT-Outline}, we use an innovation saturation mechanism  to deal with   measurement outliers, as shown in~\eqref{DT-saturated-update}. In analogy to~\eqref{sat_bound_dynamics-CT}, the saturation bound is dynamically adjusted by
\begin{subequations}\label{sat_bound_dynamics-DT}
\begin{align}\label{sat_bound_dynamics-DT-1}
 \sigma_{i,k+1} &= \lambda_{1,i} \sigma_{i,k} + \gamma_{2,i} \varepsilon_{i,k} e^{- \varepsilon_{i,k}},  \ \sigma_{i,0}>0, \\ \label{sat_bound_dynamics-DT-2}
  \varepsilon_{i,k+1} &= \lambda_{2,i} \varepsilon_{i,k} + \gamma_{2,i} \left(y_{i,k} - C_i \hat x_{k|k-1}\right)^2,  \ \varepsilon_{i,0}>0,
\end{align}
\end{subequations}
for $i=1,2,\ldots,p$, where $0 < \lambda_{1,i},  \lambda_{2,i} <1$ and $\gamma_{1,i}, \gamma_{2,i}>0$.

\subsection{Stability Analysis}
We consider  the stability analysis for the above IS-EKF when it is applied to a linear deterministic system:
\begin{align*}
\left\{
\begin{aligned}
 x_{k+1} &= A x_k,\\
y_k &= C x_k +Dd_k.
\end{aligned}
\right.
\end{align*}
Defining the state prediction error as $e_k = \hat x_{k|k-1}-x_k$, its dynamics can be expressed as
\begin{align}\label{e_k_dynamics} \nonumber
e_{k+1} &=  \hat x_{k+1|k} - x_{k+1} \\
&= A e_k - A K_k \cdot  \sat_\sigma \left(Ce_k-Dd_{k} \right)  .
\end{align} 
Before proceeding further, we show some results that will be needed later. Suppose that $(A,Q^{1\over 2})$ is stabilizable and that $(A,C)$ detectable. Then, $P_{k|k-1}$ will converge to a fixed positive definite matrix $P_\infty$ that satisfies
\begin{align*}
P_\infty &= A P_\infty A^\T +Q - A P_\infty C^\T \left( C P_\infty C^\T +R \right)^{-1}\\ &\quad \cdot C P_\infty A^\T.
\end{align*}
It is also known that $P_{k|k-1}$ is upper and lower bounded, and so is $P_{k|k}$. Hence, there should exist $\epsilon$ such that $P_{k|k}^{-1} \leq \epsilon I$. Then,
\begin{align*}
&P_{k+1|k}^{-1} = \left(A P_{k|k} A^\T+ Q\right)^{-1} \\
& = A^{-\T} \left( P_{k|k} + A^{-1} Q A^{-\T}\right)^{-1} A^{-1}\\
&= A^{-\T} \left[ P_{k|k}^{-1} - P_{k|k}^{-1}  \left( P_{k|k}^{-1}  + A^\T Q^{-1}A \right)^{-1} P_{k|k}^{-1}\right] A^{-1}\\
&\leq A^{-\T} \left[ P_{k|k}^{-1} - P_{k|k}^{-1}  \left( \epsilon I + A^\T Q^{-1}A \right)^{-1} P_{k|k}^{-1}\right] A^{-1}\\
& = A^{-\T} \left[ P_{k|k}^{-1} - P_{k|k}^{-1}  \bar Q P_{k|k}^{-1}\right] A^{-1},
\end{align*}
where $\bar Q = \left( \epsilon I + A^\T Q^{-1}A \right)^{-1}$.
We also define
\begin{align*}
Z_k = \left[ \begin{matrix} T_{1,k} - \alpha P_{k|k-1}^{-1} & T_{2,k}-C^\T W & T_{3,k} \cr
\star & T_{4,k}+2W & T_{5,k}+WD \cr
\star & \star & U \end{matrix}\right],
\end{align*}
where $W$ is a diagonal positive definite matrix, $U$ a positive definite matrix, $\alpha>0$ a positive scalar, and 
\begin{align*}
T_{1,k} &=C^\T R^{-1}C +P_{k|k}^{-1}\bar Q  P_{k|k}^{-1} - P_{k|k}^{-1}\bar Q C^\T R^{-1} C \\
&\quad - C^\T R^{-1} C \bar Q P_{k|k}^{-1} -C^\T R^{-1}C \left( P_{k|k} -\bar Q\right) \\ 
& \quad \cdot C^\T R^{-1} C - C^\T \Gamma_2 C,\\
 T_{2,k} &= - C^\T R^{-1} + P_{k|k}^{-1}\bar Q C^\T R^{-1} \\
&\quad +C^\T R^{-1}C \left(P_{k|k} - \bar Q\right) C^\T R^{-1},\\
T_{3,k} &=  \Big[- C^\T R^{-1} + P_{k|k}^{-1}\bar Q C^\T R^{-1}\Big. \\
&\quad \Big.+C^\T R^{-1}C \left(P_{k|k} - \bar Q\right) C^\T R^{-1} + C^\T \Gamma_2 \Big] D,\\
T_{4,k} &=  -R^{-1}C \left(P_{k|k} - \bar Q\right) C^\T R^{-1},\\
T_{5,k} &= -\left[ R^{-1}C \left(P_{k|k} - \bar Q\right) C^\T R^{-1} \right] D.
\end{align*}
In addition, we consider another matrix
\begin{align*}
 T_{6,k} =  D^\T \left[ R^{-1}C \left(P_{k|k} - \bar Q\right) C^\T R^{-1} + \Gamma_2 \right] D.
\end{align*}
Here, the upper boundedness of $P_{k|k}$ implies that $T_{6,k}$ is also upper bounded, with the bound denoted as $\bar {T}_{6}$.

The following theorem shows the result about the stability of the prediction error dynamics.
\begin{thm}\label{DT-SKF-Stability-Thm}
Suppose   $\|d_k\|\leq \mu <\infty$, $\rho = e^{-1}\textstyle \sum_i  \gamma_{1,i}<\infty$, and that $A$ is invertible. If there exist $P_{0|-1}$, $W$, $U$, $\alpha$ and  $\Gamma_2$ such that $Z_k\geq 0$ and $0<\alpha \leq 1- \max\{\ldots, \lambda_{1,i},\lambda_{2,i},\ldots\}$ for $i=1,2,\ldots,p$, then the prediction error $e_k$ is upper bounded with
\begin{subequations}
\begin{align}\label{DT-SKF-error-bound-1}
&\|e_k\|  \leq \sqrt{\frac{1}{c_2}\left[(1-\alpha)^k V_0 + \frac{1-(1-\alpha)^{k-1}}{\alpha} (c_1\mu^2+\rho) \right]},\\ \label{DT-SKF-error-bound-2}
&\lim_{k\rightarrow\infty} \|e_k\|  \leq \sqrt{\frac{ c_1\mu^2+\rho}{\alpha c_3}},
\end{align}
\end{subequations}
where $c_1 = \bar{\lambda}(\bar T_6+U)$, $c_2 = \underline{\lambda}(P_{k|k-1}^{-1})$, and $c_3 = \underline{\lambda}( P_\infty^{-1})$. 
Further, if $0<\alpha \leq 1- \max\{\ldots, \lambda_{1,i},\lambda_{2,i}+\gamma_{1,i},\ldots\}$ for $i=1,2,\ldots,p$, then
\begin{subequations}
\begin{align}\label{DT-SKF-error-bound-1}
&\|e_k\|  \leq \sqrt{\frac{1}{c_2}\left[(1-\alpha)^k V_0 + \frac{1-(1-\alpha)^{k-1}}{\alpha} c_1\mu^2 \right]},\\ \label{DT-SKF-error-bound-2}
&\lim_{k\rightarrow\infty} \|e_k\|  \leq \sqrt{\frac{ c_1}{\alpha c_3}}\mu.
\end{align}
\end{subequations}
\end{thm}

{\em Proof:} Consider a Lyapunov function
\begin{align*}
V_k = e_{k}^\T P_{k|k-1}^{-1} e_{k} + \textstyle \sum_i \sigma_{i,k} + \textstyle \sum_i \varepsilon_{i,k}. 
\end{align*}
Using~\eqref{e_k_dynamics}, we have
\begin{align*}
&V_{k+1}  = e_{k+1}^\T P_{k+1|k}^{-1} e_{k+1} + \textstyle \sum_i \sigma_{i,k+1}  + \textstyle \sum_i \varepsilon_{i,k+1}\\
& \leq  \left[ A e_k - A K_k \cdot \sat_\sigma (CAe_k - D d_k) \right]^\T  \\ 
&\quad  \cdot  A^{-\T}   \left[ P_{k|k}^{-1} - P_{k|k}^{-1}  \bar Q P_{k|k}^{-1} \right] A^{-1} \\
&\quad  \cdot  \left[ A e_k - A K_k \cdot \sat_\sigma (Ce_k - D d_k) \right] \\ &\quad + \left( Ce_k -Dd_k\right)^\T \Gamma_2 \left( Ce_k -Dd_k\right) +\textstyle \sum_i \lambda_{1,i} \sigma_{i,k} \\
&\quad+ \textstyle \sum_i \lambda_{2,i} \varepsilon_{i,k} +   \textstyle \sum_i  \gamma_{1,i} \varepsilon_{i,k} e^{-\varepsilon_{i,k}}\\
&\leq e_k^\T  \left[ P_{k|k}^{-1} - P_{k|k}^{-1}  \bar Q P_{k|k}^{-1} \right]  e_k - 2 e_k^\T   \left[ P_{k|k}^{-1} - P_{k|k}^{-1}  \bar Q P_{k|k}^{-1} \right] \\
&\quad \cdot K_k  \cdot \sat_\sigma (Ce_k - D d_k) + \sat_\sigma^\T (Ce_k - D d_k) \cdot K_k^\T  \\
&\quad \cdot  \left[ P_{k|k}^{-1} - P_{k|k}^{-1}  \bar Q P_{k|k}^{-1} \right]  K_k \cdot \sat_\sigma (Ce_k - D d_k)\\
&\quad + \left( Ce_k -Dd_k\right)^\T \Gamma_2 \left( Ce_k -Dd_k\right) +\textstyle \sum_i \lambda_{1,i} \sigma_{i,k} \\
&\quad+ \textstyle \sum_i \lambda_{2,i} \varepsilon_{i,k} +   \rho.
\end{align*}
Let us define
\begin{align*}
s_k = Ce_k -D d_k - \sat_\sigma(Ce_k-Dd_k).
\end{align*}
In addition, we have $P_{k|k}^{-1} = P_{k|k-1}^{-1} +C^\T R^{-1}C$, 
$P_{k|k}^{-1} K_{k } = C^\T R^{-1}$ and  $K_k^\T P_{k|k}^{-1} K_{k} = R^{-1} C P_{k|k} C^\T R^{-1}$. These relations can be readily proven. It then follows that
\begin{align*}
V_{k+1} &\leq e_k^\T P_{k|k-1}^{-1} e_k - e_k^\T T_{1,k} e_k-2 e_k^\T T_2 s_k -2 e_k^\T T_{3,k} d_k\\
& \quad - s_k^\T T_{4,k} s_k - 2 s_k^\T T_{5,k} d_k  + d_k^\T T_{6,k} d_k + \textstyle \sum_i \lambda_{1,i} \sigma_{i,k}\\
&\quad+ \textstyle \sum_i \lambda_{2,i} \varepsilon_{i,k} +  \rho.
\end{align*}
According to~\cite[Lemma 1.6]{Tarbouriech:Springer:2011}, we have
\begin{align*}
-s_k^\T W (s_k - Ce_k +D d_k) \geq 0.
\end{align*}
It can be obtained that
\begin{align*}
V_{k+1} &\leq V_{k+1} - 2s_k^\T W (s_k - Ce_k +D d_k) \\
&\leq  e_k^\T P_{k|k-1}^{-1} e_k - e_k^\T T_{1,k} e_k-2 e_k^\T (T_2-C^\T W) s_k \\
&\quad  -2 e_k^\T T_{3,k} d_k - s_k^\T \left( T_{4,k}+2W \right) s_k\\
&\quad  - 2 s_k^\T \left( T_{5,k}+WD \right) d_k  + d_k^\T T_{6,k} d_k \\
&\quad + \textstyle \sum_i \lambda_{1,i} \sigma_{i,k}+ \textstyle \sum_i \lambda_{2,i} \varepsilon_{i,k} +  \rho\\
&= e_k^\T P_{k|k-1}^{-1} e_k - \left[ \begin{matrix} e_k \cr s_k \cr d_k\end{matrix}\right]^\T \\
&\quad \cdot
\left[ \begin{matrix} T_{1,k} & T_{2,k}-C^\T W & T_{3,k} \cr
\star & T_{4,k}+2W & T_{5,k}+WD \cr
\star & \star & U \end{matrix}\right]\left[ \begin{matrix} e_k \cr s_k \cr d_k\end{matrix}\right]\\
&\quad + d_k^\T \left(T_{6,k}+U \right) d_k  + \textstyle \sum_i \lambda_{1,i} \sigma_{i,k}+ \textstyle \sum_i \lambda_{2,i} \varepsilon_{i,k} +  \rho.
\end{align*}
If $Z_k\geq 0$, then
\begin{align*}
V_{k+1}& \leq (1-\alpha ) e_k^\T P_{k|k-1}^{-1} e_k + (1-\alpha ) \textstyle \sum_i \sigma_{i,k} \\ 
&\quad + (1-\alpha ) \textstyle \sum_i \varepsilon_{i,k} +  d_k^\T \left(T_{6,k}+U \right) d_k  \\
&\quad + \textstyle \sum_i (\lambda_{1,i} +\alpha-1)\sigma_{i,k} + \textstyle \sum_i (\lambda_{2,i} +\alpha-1)\varepsilon_{i,k} +\rho. 
\end{align*}
Because  $0<\alpha \leq 1- \max\{\lambda_{1,i},\lambda_{2,i}\}$ for $i=1,2,\ldots,p$,
\begin{align*}
V_{k+1} &\leq (1-\alpha) V_k +  d_k^\T \left(T_{6,k}+U \right) d_k + \rho\\
&\leq (1-\alpha) V_k +  \bar{\lambda}\left(\bar T_6+U \right)\mu^2+\rho,
\end{align*}
from which one can easily obtain~\eqref{DT-SKF-error-bound-1}-\eqref{DT-SKF-error-bound-2}. \hfill$\bullet$

\begin{rmk}
For the discrete-time case, it is also recommended that $P_{k|k-1}$ is initialized with a small $P_{0|-1}$, which can be   near zero, in order to make the conditions in Theorem~\ref{DT-SKF-Stability-Thm} satisfied more easily. \hfill$\bullet$
\end{rmk}

\begin{figure}  \centering
    \subfigure[]{
    \includegraphics[width=0.38\textwidth]{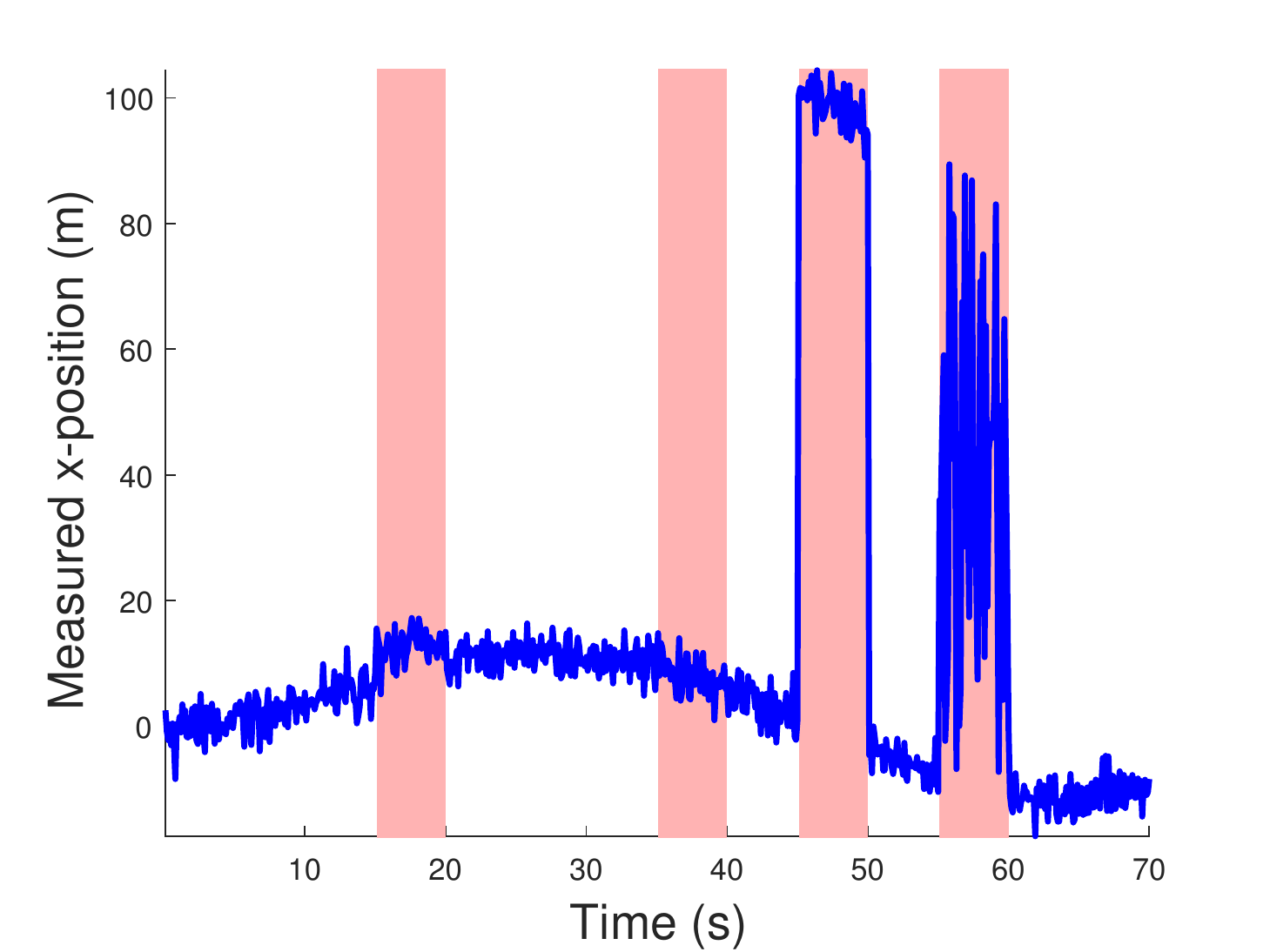}\label{Ex1-x1}}\\ 
  \subfigure[]{
    \includegraphics[width=0.38\textwidth]{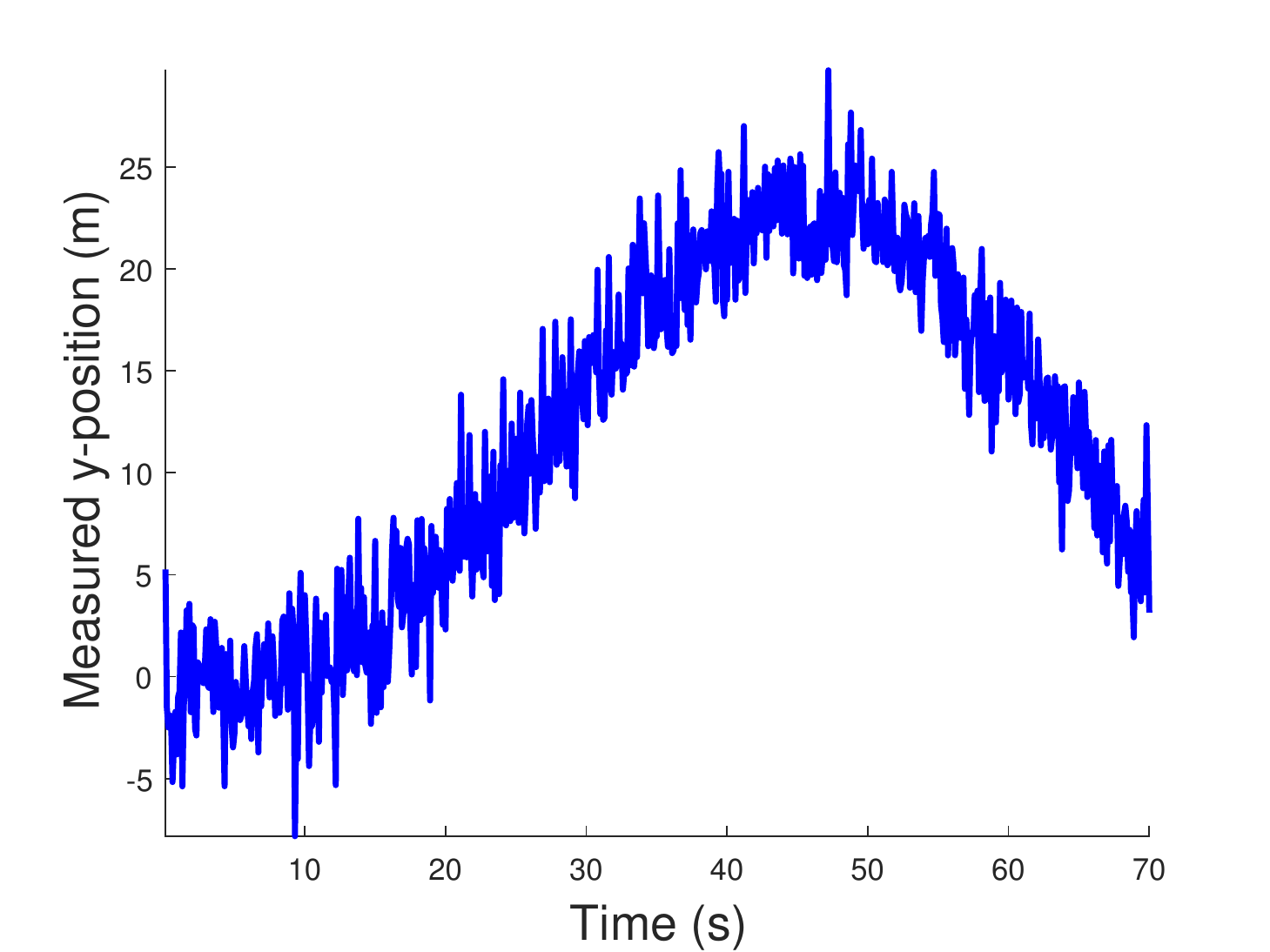}\label{Ex1-x2}}\\ 
   \subfigure[]{
   \includegraphics[width=0.38\textwidth]{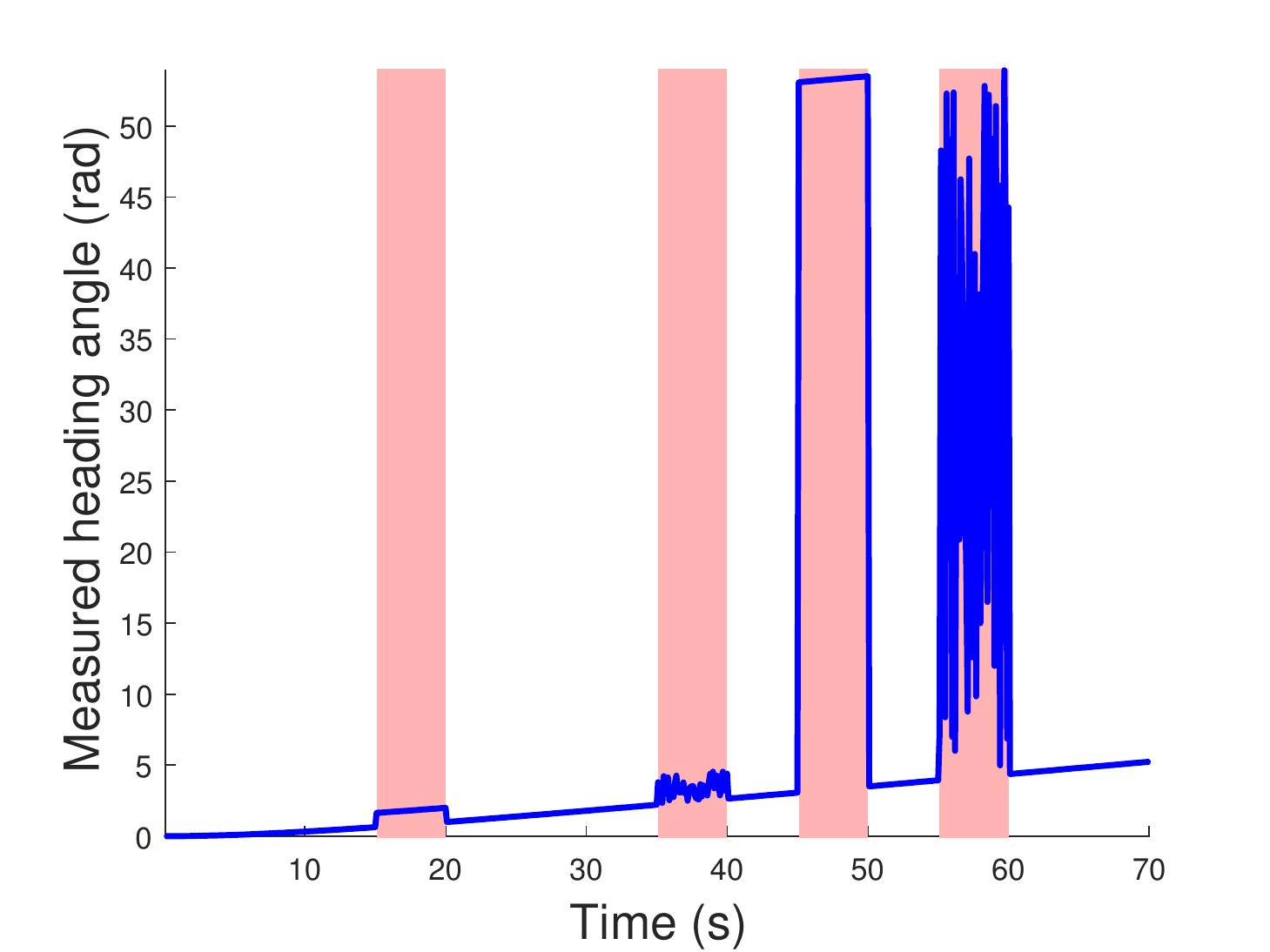}\label{Ex1-x3}}\\ 
  \caption{Measurement profiles : (a) $x$-position; (b) $y$-position; (c) heading angle. The shaded areas represent the occurrence of outliers.}
\label{Measurements}
\end{figure}

\begin{figure}  \centering
    \subfigure[]{
    \includegraphics[width=0.38\textwidth]{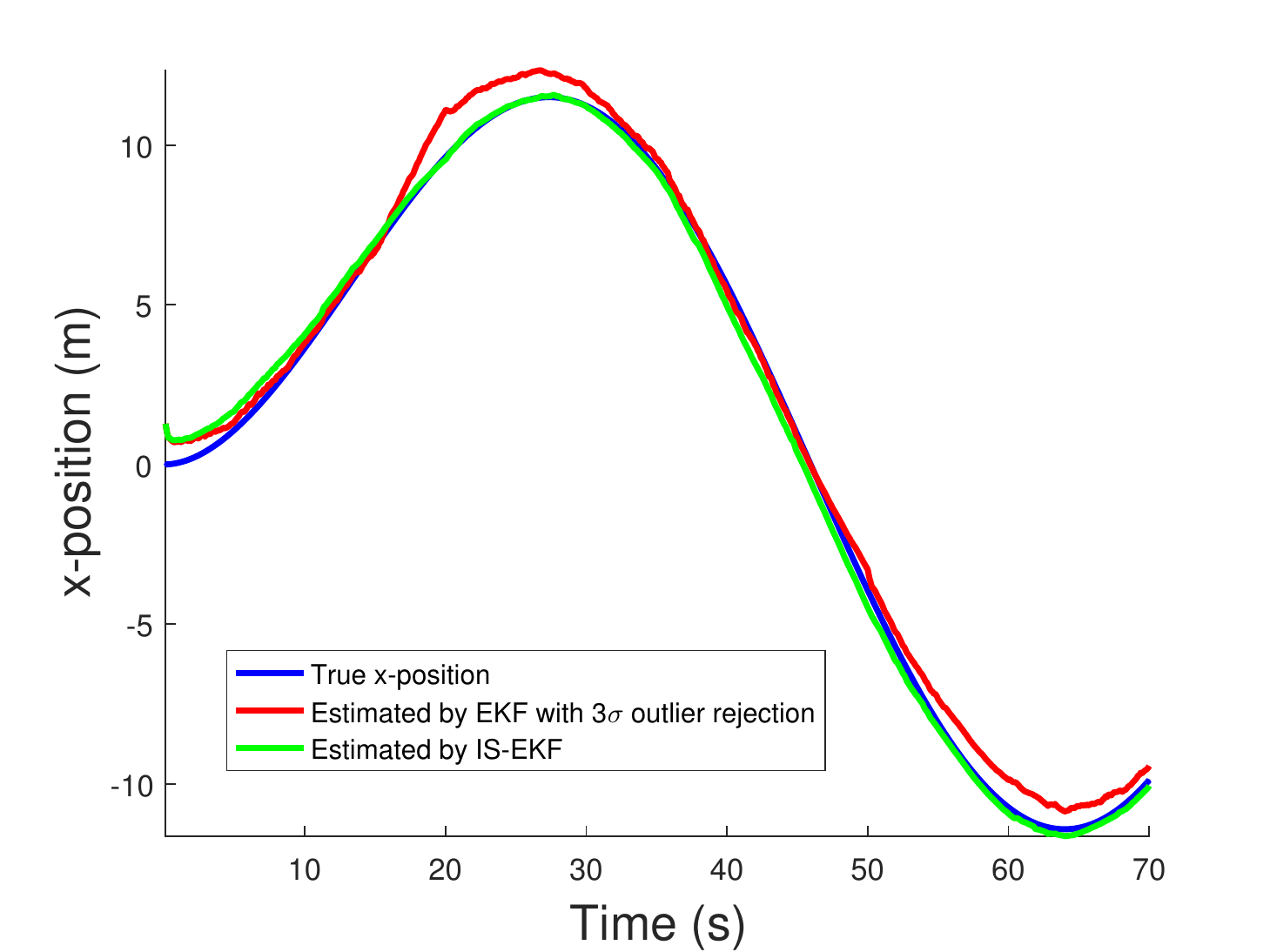}\label{Ex1-x1}}\\ 
  \subfigure[]{
    \includegraphics[width=0.38\textwidth]{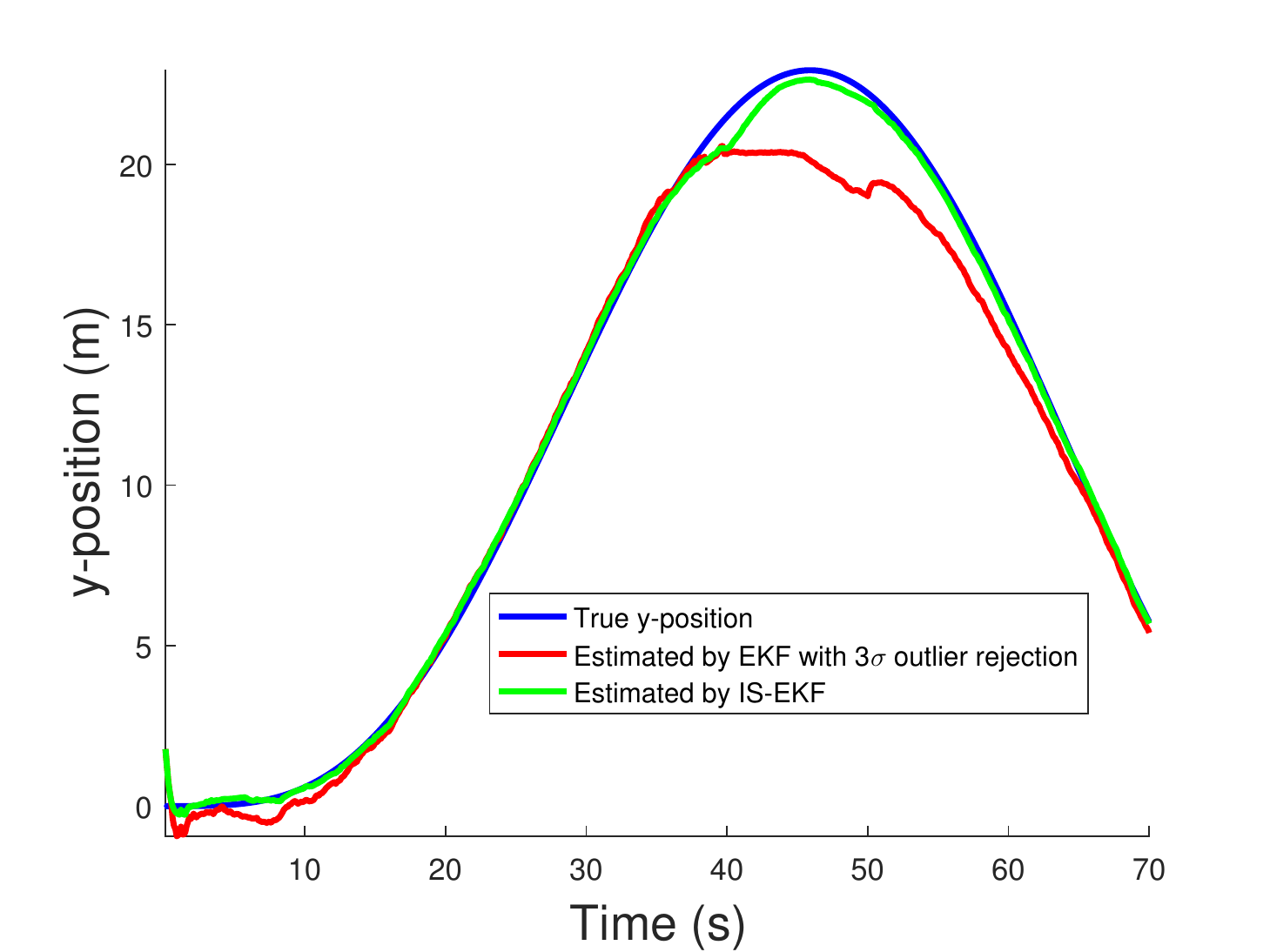}\label{Ex1-x2}}\\ 
   \subfigure[]{
   \includegraphics[width=0.38\textwidth]{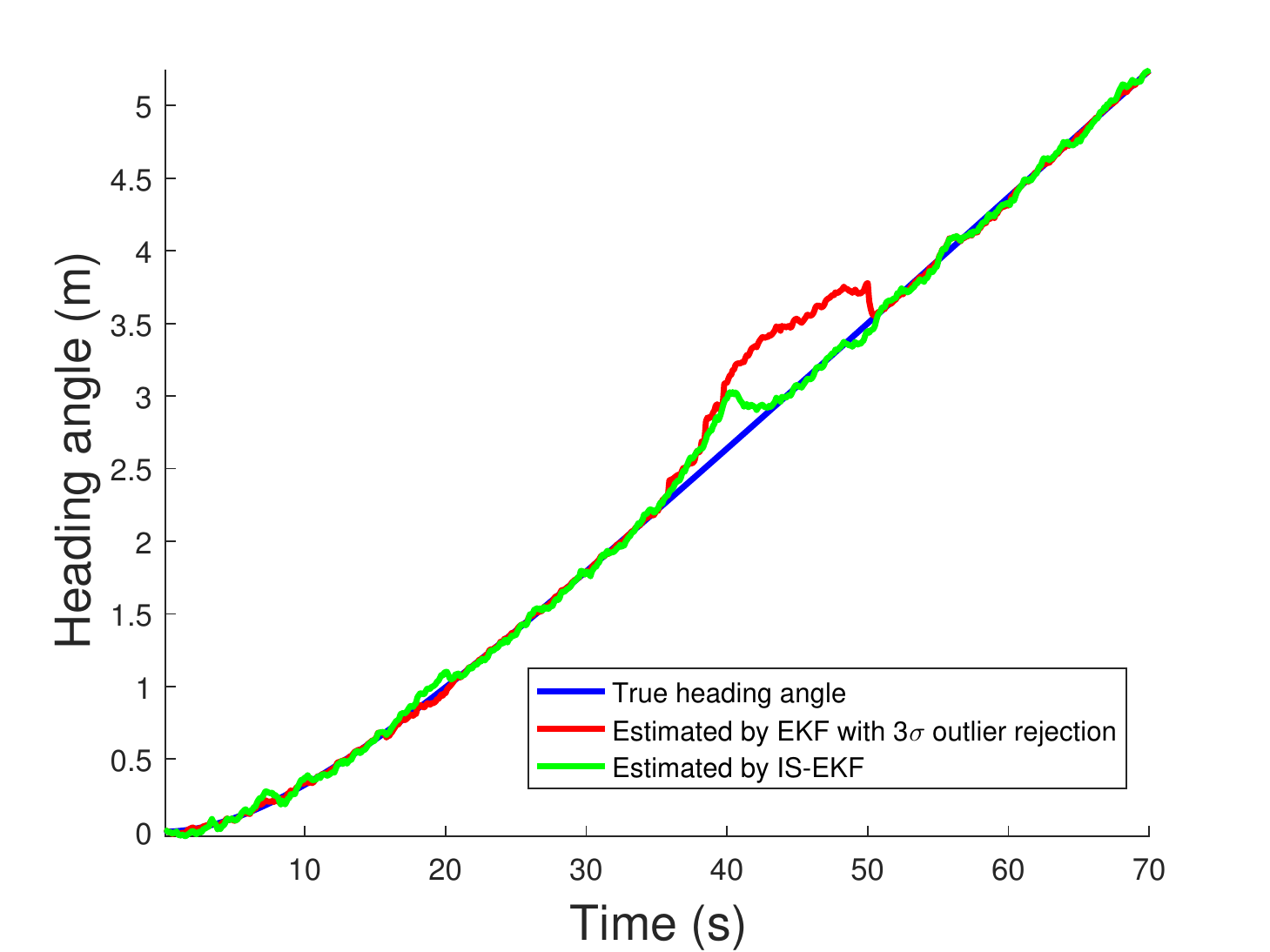}\label{Ex1-x3}}\\ 
  \caption{Estimation of the state variables: (a) estimation of  the $x$-position; (b) estimation of the $y$-position; (c) estimation of the heading angle.}
\label{x-estimation}
\end{figure}

\begin{figure}  \centering
    \includegraphics[width=0.45\textwidth]{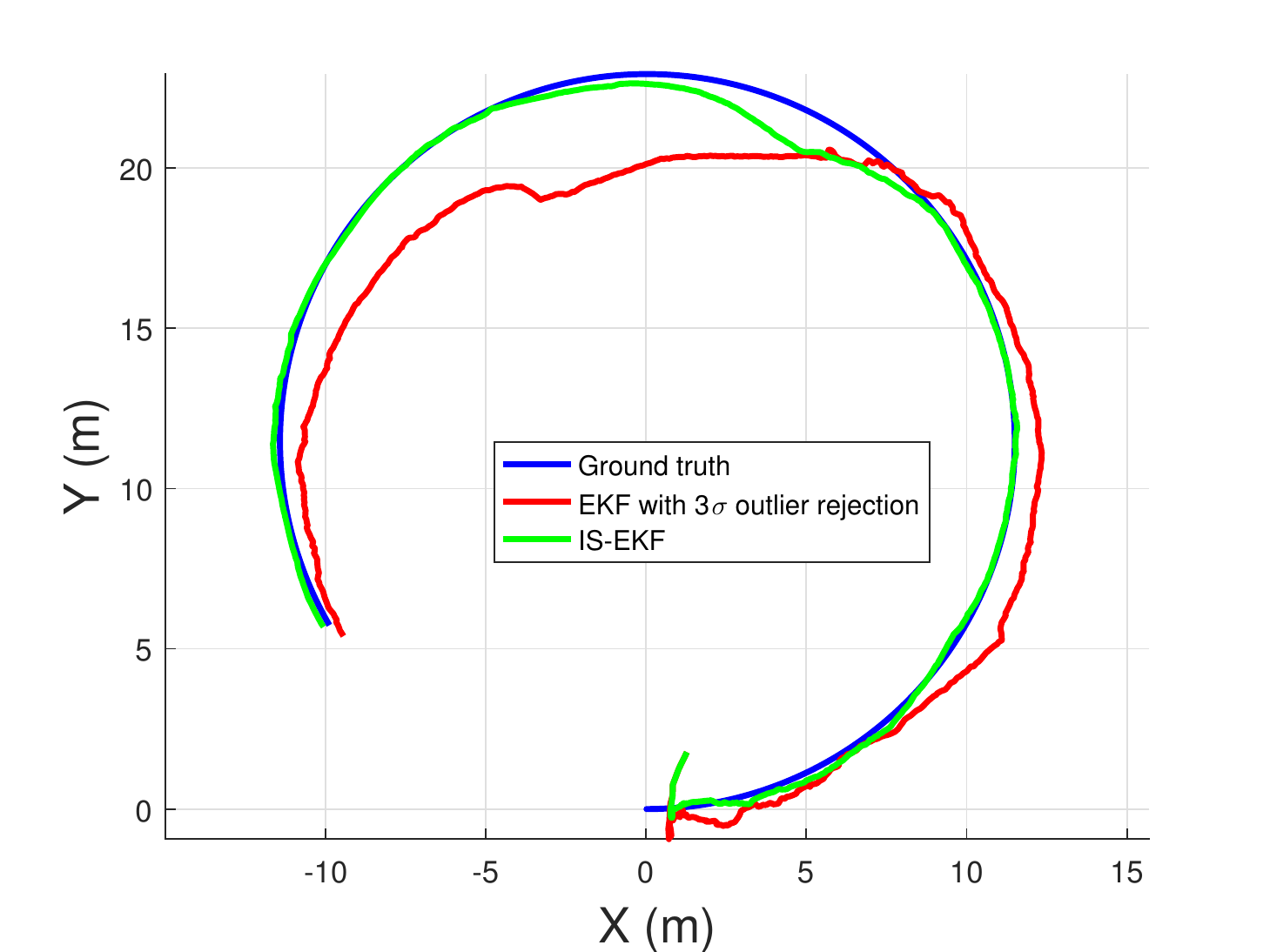}
  \caption{Estimated trajectories in comparison with the ground truth.}
\label{Estimated-Trajectory}
\end{figure}

\section{Application Example}\label{simulation}

In this section, we offer an application example to show the effectiveness of the proposed IS-EKF. The example is concerned with mobile robot localization, the purpose of which is to enable a mobile robot to determine its position and attitude using  GPS measurement data. The EKF has been a commonly used tool for this problem. However,   GPS signals are often subjected to outliers,  because they are  generally weak and can easily suffer from satellite
ephemeris, jamming, or errors in the satellite/receiver clock~\cite{Corke:2017}. This issue can significantly degrade the   accuracy of localization. 

Consider a mobile  wheeled robot with the following dynamic model~\cite{Sakai:IFAC:2010,Ghadiri-Modarres:TCST:2017}:
\begin{align*}
p_{x,k+1} &= p_{x,k} + \eta_k    T \cos(\theta_k),\\
p_{y,k+1} &= p_{y,k} + \eta_k    T \sin(\theta_k),\\
\theta_{k+1} & = \theta_k + T \delta_k,
\end{align*}
where $p_{x,k}$ and $p_{y,k}$ are the coordinates of the robot's center of mass, $\theta_k$  the heading angle, $\eta_k$ the robot's speed at the center of mass, $\delta_k$ the steering angle, and $T$ the sampling period. Thus, the state vector of the robot is $\left[ p_{x,k} \ p_{y,k} \ \theta_k \right]^\T$, and the control input vector is $\left[ \eta_k \ \delta_k \right]^\T$.  Here, $\eta_k$ and $\delta_k$ can be  read from onboard meters. 
In addition, we let $x_k$ and $y_k$   be obtained using GPS and   $\theta_k$ be measured by a compass.  The measurement model hence is a linear equation, but the dynamic model is nonlinear. Process and measurement noises are included into the model to represent uncertainties.

When the robot is moving and sampled every $T=0.1$ s,   the measurements of the $x$-coordinate and the steering angle are corrupted by outliers from time to time. The outlier disturbance $d_k$ is designed as
\begin{align*}
d_k = \left\{ \begin{array}{ll} \left[ \begin{matrix} 5 \ 1 \end{matrix} \right]^\T  & 150< k \leq 200,\\
2 \zeta_k  & 350< k \leq 400,\\
\left[ \begin{matrix} 100 \ 50 \end{matrix} \right]^\T  & 450< k \leq 500,\\
\left[ \begin{matrix} 100 & 0 \cr 0 & 50 \end{matrix} \right]\zeta_k  & 550< k \leq 600,
  \end{array} \right.
\end{align*}
where  $\zeta_k \in \mathbb{R}^2$ is a uniform random vector. This design takes account of different types of outliers in four stages --- $d_k$    is small at Stages 1-2 and large Stages 3-4, and it is constant at Stages 1 and 3 and random at Stages 2 and 4. In this setting, the synthetic measurements are generated and shown in Fig.~\ref{Measurements}, in which outlier-corrupted measurements are displayed in shaded areas. 

It can be easily verified that the conventional EKF will completely fail when the above outliers are imposed on the measurements. One way to improve the robustness of the EKF is to use the $ 3\sigma$ rule~\cite{Liu:CCE:2004}. That is,  the innovation $(y_k - h(\hat x_k))$ is considered as normal if it lies within the $\pm3\sigma$ bounds based on the covariance matrix $(H_k P_{k|k-1} H_k^\T+R)$, and outlying  otherwise. It is reset to be zero in the latter case so as not to distort the update procedure. We call this method as EKF with $3\sigma$ outlier rejection and use it to compare with the IS-EKF. In addition, to apply the IS-EKF in~\eqref{SKF-DT-Outline}-\eqref{sat_bound_dynamics-DT}, the following parameters are set for~\eqref{sat_bound_dynamics-DT}:
\begin{align*}
\Lambda_1 = \mathrm{diag}(\left[ \begin{matrix} 0.5 & 0.5 & 0.1  \end{matrix} \right]), \ 
\Lambda_2 = \mathrm{diag}(\left[ \begin{matrix} 0.1 & 0.1 & 0.1  \end{matrix} \right]),\\
\Gamma_1 = \mathrm{diag}(\left[ \begin{matrix} 10^2 & 10^2 & 5\times 10^{-3}  \end{matrix} \right]), \
\Gamma_2 = \mathrm{diag}(\left[ \begin{matrix} 9 & 9 & 9  \end{matrix} \right]).
\end{align*}

 Figs.~\ref{Ex1-x1}-\ref{Ex1-x3} show  the estimation of the $x$- and $y$-coor-dinates and the heading angle through time, respectively. It is seen that, with the $3\sigma$ outlier rejection, the EKF do not diverge seriously but  still can not provide reliable state estimation. As a contrast, the IS-EKF  demonstrates much better estimation performance, maintaining a smooth and accurate estimation  when the outliers appear. Further, Fig.~\ref{Estimated-Trajectory} illustrates the estimated  trajectories in comparison with the ground truth, in which the IS-EKF reconstructs the  trajectory at a good accuracy. These results reflect a considerable effectiveness of the IS-EKF in addressing  the outliers.
\vspace{3mm}

Here, we provide some further remarks about the advantages and application of the IS-EKF. 

\begin{rmk}
In addition to the example presented above, we performed numerous simulations about robot localization in different settings (e.g., outliers, noises, initial estimation), and we found that the IS-EKF consistently offer satisfactory estimation. Two important observations include:

\begin{itemize}

\item The IS-EKF can well handle outliers that last for a relatively long period and vary in magnitude or type. This contrasts with many existing methods. For example, the stubborn observer in~\cite{Alessandri:AUTO:2018} can only reject occasional, singly outliers, and the EKF with $3\sigma$ outlier rejection is less effective for dealing with small outliers that approximately align with the $\pm 3 \sigma$ bounds. Note that the EKF may offer improved estimation by adjusting the  outlier detection bounds from $\pm3\sigma$ to $\pm \ell \sigma$ with $\ell \in \mathbb{R}^+$. However, for each selection of $\ell$, it can still be futile for  outliers that roughly lie within the bounds. 

\item As an additional benefit, the IS-EKF  demonstrates good robustness again initial state uncertainties. It is known that the EKF can easily fail if the initial guess is not accurate, although it is often difficult to obtain a precise guess in practice. However, the innovation saturation can often check the divergence caused by a poor initial guess, reducing the possibility of a complete failure. \hfill$\bullet$
\end{itemize}
\end{rmk}

\begin{rmk}
Application of the IS-EKF requires the selection of a set of parameters for  the innovation saturation procedure. We have the following suggestions for practitioners:
\begin{itemize}

\item Let $\lambda_{1,i},\lambda_{2,i}<0$ for a continuous-time system and $0<\lambda_{1,i},\lambda_{2,i}<1$ for a discrete-time system. 

\item It is useful to choose $\lambda_{2,i}$ such that the dynamics $\varepsilon_i$ is fast in order to better track the change in innovation. This can be done by letting it take an appropriately small negative number in the continuous-time case or  a number close to zero in the discrete-time case.

\item Let $\gamma_{1,i},\gamma_{2,i}>0$ and $\gamma_{2,i}<10$.

\item The selection of $\gamma_{1,i}$ depends on the magnitude of the corresponding  innovation process. Set $\gamma_{1,i}$  to be a large number if the innovation is usually large in the normal, i.e., outlier-free, case and a small number  otherwise.  \hfill$\bullet$

\end{itemize}
\end{rmk}

\section{Conclusion}\label{conclusion}

The EKF has gained wide application across different fields as a popular state estimation tool. However, its estimation accuracy can be severely hampered by measurement outliers due to sensor anomaly, model uncertainties, data transmission errors or cyber attacks. This paper presented a novel innovation saturation mechanism, IS-EKF, which is a robustified EKF, as an alternative to the conventional EKF.  This mechanism saturates the innovation process, which is crucial for correcting the state estimation, thus ensuring a reasonable correction to be applied when outliers occur. We showed the IS-EKF architecture  for both continuous- and discrete-time systems and provided theoretical analysis to derive useful stability properties of the proposed approaches for linear systems. We conducted an application study on  mobile robot localization  to illustrate  the effectiveness of IS-EKF architecture. The numerical simulation  results  showed that the proposed approach brings about significant robustness for localization against GPS outliers.  This shows the viability of IS-EKF in effectively rejecting outliers of varying magnitude and durations at a reasonable computational cost and without the need of measurement redundancy.

\balance

\bibliographystyle{IEEEtran} 
\bibliography{Saturated-EKF-TCST_V4} 
\end{document}